\documentclass[11pt]{article}
\usepackage{color,graphicx}

\usepackage{amsmath,amsfonts,amssymb,amsthm,placeins,float,datetime}

\begin{document}

\title{Kim-type APN functions are affine equivalent to Gold functions}
\author{Benjamin Chase\\ Petr Lison\v{e}k\thanks{Research of both authors was supported
in part by the Natural Sciences and Engineering Research Council of Canada
(NSERC).}
\\
Department of Mathematics\\
Simon Fraser University\\
Burnaby, BC\\
Canada\ \ \ V5A 1S6\\
\ \\
{\tt plisonek@sfu.ca}}
\date{}

\def\eps{\varepsilon}
\def\F{{\mathbb F}}
\def\Q{{\mathbb Q}}
\def\C{{\mathbb C}}
\def\Z{{\mathbb Z}}
\def\Fq{{\mathbb F_q}}
\def\e{\eta}
\def\L1{L(1,z^*)}
\def\Le{L(\e,z^*)}

\def\P{\phantom{\Big|}}

\newtheorem{theorem}{Theorem}[section]
\newtheorem{lemma}[theorem]{Lemma}
\newtheorem{proposition}[theorem]{Proposition}
\newtheorem{corollary}[theorem]{Corollary}
\newtheorem{conjecture}[theorem]{Conjecture}
\theoremstyle{definition}
\newtheorem{definition}[theorem]{Definition}
\newtheorem{example}[theorem]{Example}
\newtheorem{remark}[theorem]{Remark}
\newtheorem{problem}[theorem]{Problem}

\def\ra{\rightarrow}

\def\gf{\F}
\def\tr{{\rm Tr}}
\def\Tr{{\rm Tr}}

\def\t{\theta}

\def\w{\omega}

\maketitle

\begin{abstract}
	The problem of finding APN permutations of $\F_{2^n}$
	where $n$ is even and $n>6$ has been called the Big APN Problem.
Li, Li, Helleseth and Qu recently characterized APN functions
defined on $\F_{q^2}$
of the form $f(x)=x^{3q}+a_1x^{2q+1}+a_2x^{q+2}+a_3x^3$,
where $q=2^m$ and  $m\ge 4$.
We will call functions of this form
Kim-type functions because they generalize the form
of the Kim function that was used
to construct an APN permutation of $\F_{2^6}$.
We extend the result of Li, Li, Helleseth and Qu by proving that if
a Kim-type function $f$ is APN and $m\ge 4$,
then $f$ is affine equivalent to one of two
Gold functions $G_1(x)=x^3$ or $G_2(x)=x^{2^{m-1}+1}$.
Combined with the recent result
of G\"{o}lo\u{g}lu and Langevin
who proved that, for even $n$, Gold APN functions
are never CCZ equivalent to permutations,
it follows that 
for $m\ge 4$
Kim-type APN functions 
on $\F_{2^{2m}}$
are never CCZ equivalent to permutations.
\end{abstract}

\section{Background}

Let $\F_{2^k}$ denote the finite field with $2^k$ elements
and let $\F_{2^k}^*$ be the set of its nonzero elements.
Let $\Tr_k$ denote the absolute trace from $\F_{2^k}$ to $\F_{2}$.

We say that a function $f:\F_{2^n}\rightarrow \F_{2^n}$
is {\em almost perfect nonlinear (APN)}
if for any $a\in\F_{2^n}^*$ and any $b\in\F_{2^n}$
the equation
\[
f(x+a)+f(x)=b
\]
has at most two solutions $x\in\F_{2^n}$.
We say that functions $f:\F_{2^n}\rightarrow \F_{2^n}$
and $g:\F_{2^n}\rightarrow \F_{2^n}$
are {\em extended affine (EA) equivalent}
if there exist affine permutations $L_1,L_2$
of $\F_{2^n}$ and an affine mapping
$L_3:\F_{2^n}\rightarrow \F_{2^n}$
such that $g(x)=L_1(f(L_2(x)))+L_3(x)$.
If $L_3(x)=0$ we will say that
$f$ and $g$ are {\em affine equivalent;}
this is the kind of equivalence that we will be using
in the present paper.
EA equivalence is a special case
of {\em Carlet-Charpin-Zinoviev (CCZ) equivalence} \cite{CCZ}.
CCZ equivalence preserves 
the property whether the function is APN or not.

Many infinite families of APN functions are known.
A recent account of the known families
is given in \cite{BCV}.
In the present paper we will encounter
the {\em Gold} APN functions
$G(x)=x^{2^k+1}$ where $\gcd(k,n)=1$.

APN functions are of great interest in cryptography as components of S-boxes
in block ciphers,
where they provide optimal protection against differential
cryptanalysis. They also have connections to combinatorial
objects such as finite geometries and combinatorial designs
\cite{Lietal}.

In some designs of block ciphers, such as
substitution-permutation networks, it is desirable that 
the APN function is also a {\em permutation} of $\F_{2^n}$.
While many APN permutations are known
for odd $n$, the situation is very different 
when $n$ is even. It is known that no 
APN permutations of $\F_{2^n}$ exist when
$n=2,4$.
An APN permutation of $\F_{2^6}$
was announced
at the Fq9 conference in 2009 \cite{BDMW}.
Its construction is based
on the APN function 
\[
\kappa(x)=x^3+x^{10}+ux^{24}
\]
where $u$ is a suitable primitive element of $\F_{2^6}$.
The function $\kappa$ is known as {\em Kim function.}
APN permutations of $\F_{2^6}$ are then
constructed by applying a suitable
CCZ equivalence transformation
to the Kim function  \cite{BDMW};
all examples constructed
in \cite{BDMW} are mutually CCZ equivalent.

The question of existence of APN permutations of $\F_{2^n}$ 
for even $n>6$ was posed as the {\em Big APN Problem} in \cite{BDMW},
and it has attracted  lot of interest since then.

By generalizing the form of the Kim function, Carlet in \cite{carlet2014open} posed the following open problem:

\begin{problem}\label{Prob1}
	\cite[Section 3.7]{carlet2014open}
	Find more APN functions or, better, infinite classes of APN functions of the form $X^3+aX^{2+q}+bX^{2q+1}+cX^{3q}$ where $q=2^{n/2}$ with $n$ even, or more generally of the form  $X^{2^k+1}+aX^{2^k+q}+bX^{2^kq+1}+cX^{2^kq+q},$ where $\gcd(k,n)=1$. 
\end{problem} 

The motivation for solving this problem is given by possibly applying
a suitable CCZ transformation to the resulting APN functions
to obtain APN permutations, in analogy to the approach of \cite{BDMW}.
There have been several lines of attack to approach
Problem~\ref{Prob1} and similar problems; we refer
to the Introduction section of \cite{LLHQ} for a detailed account of them.

Most recently, Li, Li, Helleseth and Qu \cite{LLHQ} solved the first part
of Problem~\ref{Prob1} by proving Theorem~\ref{thm-LLHQ} stated below,
in which they completely characterize APN functions
on $\F_{2^n}$
of the form $f(x)=x^{3q}+a_1x^{2q+1}+a_2x^{q+2}+a_3x^3$,
where $q=2^m$ and $m=n/2$, $m\ge 4$.
We will call functions of this form
{\em Kim-type functions} in the present paper.
We extend the result of \cite{LLHQ} by proving that if
a Kim-type function $f$ is APN and $m\ge 4$,
then $f$ is affine equivalent to one of two
Gold functions $G_1(x)=x^3$ or $G_2(x)=x^{2^{m-1}+1}$.
Combined with the recent result
of G\"{o}lo\u{g}lu and Langevin \cite{GL}
which states that, for even $n$, Gold APN functions
are never CCZ equivalent to permutations,
it follows that 
for $m\ge 4$
Kim-type functions are never CCZ equivalent to permutations.

\subsection{Preliminaries}

Throughout the paper we assume that $m\ge 4$ is an integer
and $q=2^m$. For $z\in\F_{q^2}$ we denote $\bar z=z^q$. 
The set
\[
U=\{z\in\F_{q^2}\;|\;z^{q+1}=1\}
\]
will play an important role. Since $z^{q+1}=z\bar{z}$
is the norm of $z\in\F_{q^2}$ with respect to $\F_q$, the set $U$ is sometimes called
the ``unit circle'' in $\F_{q^2}$. A well known fact  is that every $z\in\F_{q^2}^*$
can be written {\em uniquely} as $z=xy$ where $x\in\F_{q}$ and $y\in U$.
Of course, for $z\in U$ we have $z^q=\frac{1}{z}$.
This computational rule will come handy in many proofs in this paper.
Some of the calculations in this paper, in particular
factorization of polynomials and the computation
of a resultant of polynomials, were done with the support
of the computer algebra system Magma \cite{mag}.

We will need a condition for solubility of certain quadratic equations
by elements of $U$ distinct from~$1$.

\begin{lemma}
	\label{lem-quad-U}
	Assume that $A,B,C\in\F_q$ and the equation $Az^2+Bz+C=0$ has a solution
	$z\in U\setminus\{1\}$. Then $A=C$.
\end{lemma} 
\begin{proof}
	Suppose $Az^2+Bz+C=0$ and $z\in U\setminus\{1\}$. After raising the
	equation to the $q$-th power and simplifying we get
	$A+Bz+Cz^2=0$. Adding the equations yields $(A+C)(z^2+1)=0$. Since 
	$z\neq 1$, the conclusion follows.
\end{proof}

\subsection{Characterizations of Kim-type APN functions}
\label{sec-Kim-APN}

In Section~\ref{sec-Kim-APN} we summarize the most important
results of the recent preprints \cite{LLHQ} and \cite{GKL}.

We start with an observation which is not stated
as a numbered item in \cite{LLHQ} but we want to give it
a label for our future reference, and we also state it in
a slightly different form.

\begin{lemma}
	\label{lem-a1-small-field}
	Any function $f:\F_{q^2}\ra\F_{q^2}$  
	given by $f(x)=x^{3q}+a_1x^{2q+1}+a_3x^3$
	where $a_1,a_3\in\F_{q^2}$
	is affine equivalent to 
	a function $f'(x)=x^{3q}+a_1'x^{2q+1}+a_3'x^3$
	where $a_1'\in\F_{q}$ and $a_3'\in\F_{q^2}$.
\end{lemma}
\begin{proof}
	This is proved at the bottom of page 2 in \cite{LLHQ}.
	The proof given there applies to the general case when
	$f(x)$ also contains the term $a_2x^{q+2}$. Denote the function
	obtained in the proof by $f'(x)=x^{3q}+a_1'x^{2q+1}+a_2'x^{q+2}+a_3'x^3$.
	It follows from the proof that if $a_2=0$, then $a_2'=0$.
\end{proof}

The following constants are introduced in \cite{LLHQ}:
\begin{equation}
\label{theta}	
\left\{
\begin{array}{lr}
\theta_1 = 1+a_1^2+a_2\bar{a}_2+a_3\bar{a}_3  \\ 
\theta_2 = a_1+\bar{a}_2a_3 \\
\theta_3 = \bar{a}_2+a_1\bar{a}_3 \\
\theta_4 = a_1^2+a_2\bar{a}_2.
\end{array}
\right.
\end{equation}


The following theorem is the main result of \cite{LLHQ}.

\begin{theorem}\cite[Theorem 1.2]{LLHQ}
	\label{thm-LLHQ}
	Let $n=2m$ with $m\ge4$ and $f(x) = \bar{x}^3 + a_1\bar{x}^2x+a_2\bar{x}x^2+a_3x^3$, where $a_1\in\gf_{2^m}, a_2,a_3\in\gf_{2^n}$. 
	Let $\theta_i$'s be defined as in \eqref{theta} and 
	define
	\begin{equation}\label{Eq-4}
	\Gamma_1 = \left\{ (a_1,a_2,a_3) ~|~ \theta_1\neq0, ~ \tr_{m}\left(\frac{\theta_2\bar{\theta}_2}{\theta_1^2}\right) =0,~ \theta_1^2\theta_4 + \theta_1\theta_2\bar{\theta}_2 + \theta_2^2\theta_3 + \bar{\theta}_2^2\bar{\theta}_3 =0 \right\}
	\end{equation}
	and 
	\begin{equation}\label{Eq-5}
	\Gamma_2 = \left\{ (a_1,a_2,a_3) ~|~ \theta_1\neq0, ~\tr_{m}\left(\frac{\theta_2\bar{\theta}_2}{\theta_1^2}\right) =0, ~\theta_1^2\theta_3 + \theta_1\bar{\theta}_2^2 + \theta_2^2\theta_3 + \bar{\theta}_2^2\bar{\theta}_3 =0 \right\}.
	\end{equation}
	Then $f$ is APN over $\gf_{2^n}$ if and only if 
	\begin{enumerate} 
		\item[(1)] $m$ is even, $(a_1,a_2,a_3)\in\Gamma_1\cup\Gamma_2$; or
		\item[(2)] $m$ is odd, $(a_1,a_2,a_3)\in\Gamma_1$.
	\end{enumerate}
\end{theorem}

Prior to the appearance of the manuscript \cite{LLHQ},
a characterization of the Kim-type APN functions
for the special case when
the coefficients $a_i$ lie in the subfield $\F_q$
was given in \cite{K},
and then the following
result on equivalence with Gold functions was proved
in \cite{GKL}.

\begin{theorem} \cite{GKL}
	\label{thm-small-field}
	Suppose that $m\ge 4$ is an integer and let $q=2^m$. Let $f:\F_{q^2}\ra\F_{q^2}$ be 
	given by $f(x)=x^{3q}+a_1x^{2q+1}+a_2x^{q+2}+a_3x^3$
	where $a_1,a_2,a_3\in\F_{q}$. If $f$ is APN, then
	$f$ is affine equivalent to $G_1(x)=x^3$ or
	$f$~is affine equivalent to $G_2(x)=x^{2^{m-1}+1}$.
\end{theorem} 

In order to make the present paper self-contained,
in the remainder of Section~\ref{sec-Kim-APN}
we include the proof of Theorem \ref{thm-small-field} 
using Theorem~\ref{thm-LLHQ} and some ideas of \cite{GKL}.


Assume that $a_1,a_2,a_3\in\F_{q}$. Then $\t_1=(1+a_1+a_2+a_3)^2$
and equations (\ref{Eq-4}) and (\ref{Eq-5}) factorize as 
$\t_1 S^2=0$ and $\t_1 S T =0$ respectively,
where
\[
S = a_1^2+a_1a_3+a_2^2+a_2
\]
and
\[
T = a_1a_3+a_2+a_3^2+1.
\]
If $f$ is APN,
then $\t_1\neq 0$, thus we must have
$S=0$ for equation (\ref{Eq-4}) to hold,
and we must have
$ST=0$ for equation (\ref{Eq-5}) to hold.


Recall that a linear map 
from $\F_{q^2}$ to $\F_{q^2}$
is a bijection if and only if it has no nonzero roots in $\mathbb{F}_{q^2}$. In the following proof of Theorem~\ref{thm-small-field} we will often use linear maps of the form $L(x) = x^q + tx$ which are bijections if and only if $t\notin U$. 
For the case of $t\in\mathbb{F}_q$ we only need to check that $t\ne 1$ for $L$ to be a bijection. 

First we will show that if the function $f$ given in Theorem~\ref{thm-small-field} is APN, then $f$ is affine equivalent to 
a function that has
one of the following forms: 
\begin{itemize}
	\item[(i)] $f_1(x) = x^{3q} + c_1x^{2q+1} + c_2x^{q+2}$,
	\item[(ii)] $f_2(x) = x^{3q} + c_2x^{q+2} + x^3$,
\end{itemize} with $c_1, c_2\in \mathbb{F}_q$.

If $a_3 = 0$ then we are done. Assume $a_3 = 1$. If $a_1 = 0$ then we are done. If $a_1 = a_2$ then $f(x)\in\mathbb{F}_q$ for all $x\in\mathbb{F}_{q^2}$, which contradicts $f$ being APN. If $a_1\notin \{ 0,a_2 \}$ then let $r=a_2/a_1$,
and let 
\[
L(u)=\frac{u^q + ru}{r+1}.
\]
The function $f_2(x)=L(f(x))$ is affine equivalent
to $f$, and $f_2$ is of the form (ii) given above.

Now assume $a_3\notin\{0,1\}$. 
Let 
\[
L(u)=\frac{a_3u^q + u}{a_3^2+1}.
\]
The function $f_1(x)=L(f(x))$ is affine equivalent
to $f$, and $f_1$ is of the form (i) given above.
Therefore, if $f$ is APN, it must be affine equivalent to a function of the form (i) or (ii).

Now we will show that APN functions of the form $f_1$ or $f_2$ are affine equivalent to $G_1$ or $G_2$. First consider $f_1$. Suppose 
$$S = c_1^2 + c_2^2 + c_2 = 0.$$ 
If $c_2 = 0$ then $c_1 = 0$ and $f_1(x)^q = G_1(x)$. If $c_2 = 1$ then $S=c_1^2=0$, but by Theorem~\ref{thm-LLHQ}, we require that $\theta_1 = c_1^2\ne 0$. 
So assume $c_2\notin\{0,1\}$. Let $b\in\mathbb{F}_q\setminus\mathbb{F}_2$ be such that $c_2 = (b+1)^2$. Then it follows from $S=0$ that $c_1 = b^2 + b$. So $f_1$ is of the form
\begin{equation}
\label{eq-f1}
f_1(x) = x^{3q} + (b^2+b)x^{2q+1} + (b+1)^2x^{q+2}. 
\end{equation}
Let $r\in\mathbb{F}_q\setminus\mathbb{F}_4$ so that $r^3\ne 1$ and the linear maps $L_1(x) = x^q + rx$ and $L_2(x) = r^3x^q + x$ are bijections. We have $$L_2(G_1(L_1(x))) = (r^6+1)x^{3q} + (r^5+r)x^{2q+1} + (r^4+r^2)x^{q+2}.$$ Dividing by $(r^6+1)$, we see that $G_1$ is affine equivalent to $$F_1(x) = x^{3q} +  \frac{r^5+r}{r^6+1} x^{2q+1} +  \frac{r^4+r^2}{r^6+1} x^{q+2}.$$ Letting $$b' = \frac{(r+1)^3}{r^3+1},$$ we have $$F_1(x) = x^{3q} + (b'^2 + b')x^{2q+1} + (b'+1)^2x^{q+2}.$$ From the trace condition of Theorem~\ref{thm-LLHQ}
applied to function $f_1$ in equation (\ref{eq-f1}) we have $$\Tr_m\left(\frac{\theta_2^{q+1}}{\theta_1^2}\right) = \Tr_m\left( \frac{\theta_2}{\theta_1}\right) =
 \Tr_m\left( \frac{b^2+b}{b^2} \right) = 0,$$ 
 equivalently
  $$\Tr_m\left( \frac{1}{b} \right) = \Tr_m(1).$$ To show that $f_1$ is affine equivalent to $G_1$, we must show that $b'$ takes all values from $\mathbb{F}_q\setminus\mathbb{F}_2$ such that $\Tr_m(b'^{-1}) = \Tr_m(1)$ as $r$ runs through $\mathbb{F}_q\setminus\mathbb{F}_4$. Let $r' = r+1$. Then as $r$ runs through $\mathbb{F}_q\setminus\mathbb{F}_4$, so does $r'$, and \begin{align*}
\Tr_m(b'^{-1}) = \Tr_m\left( \frac{r^3+1}{(r+1)^3} \right) 
&= \Tr_m\left(\frac{r'^3 + r'^2 + r'}{r'^3} \right)\\
&= \Tr_m(1) + \Tr_m\left(\frac{1}{r'}\right) + \Tr_m\left(\frac{1}{r'^2}\right) = \Tr_m(1).\end{align*}
Therefore, if $S=0$,
then $f_1$ is affine equivalent to $G_1$. 

Now suppose $S\ne 0$, that is, $m$ is even and $T=c_2+1=0$. Then $c_2 = 1$ and $$f_1(x) = x^{3q} + c_1x^{2q+1} + x^{q+2}.$$ By Theorem~\ref{thm-LLHQ}, we require $\theta_1 = c_1^2\ne 0$, that is, $c_1\ne 0$. We also require the trace condition $$\Tr_m\left( \frac{\theta_2^{q+1}}{\theta_2^2} \right) = \Tr_m\left( \frac{1}{c_1} \right) = 0.$$ Let $r\in \mathbb{F}_q\setminus \mathbb{F}_2$ and let $L(x) = x^q + rx$. We have 
$$L(G_2(L(x)^{2q})) = (r^3 + r)x^{3q} + (r + 1)^4x^{2q + 1} + (r^3 + r)x^{q+2}.$$
After dividing by $r^3 + r$, we see that $G_2$ is affine equivalent to $$F_1(x) = x^{3q} + dx^{2q+1} + x^{q+2}$$ where $$d = \frac{(r+1)^2}{r}.$$ Let $r' = r+1$ so that $r'$ is also in $\mathbb{F}_q\setminus\mathbb{F}_2$. We have $$\Tr_m\left(\frac{1}{d}\right) = \Tr_m\left(\frac{r}{(r+1)^2}\right) = \Tr_m\left(\frac{r'+1}{r'^2}\right) = \Tr_m\left(\frac{1}{r'}\right) + \Tr_m\left(\frac{1}{r'^2}\right) = 0$$ 
which means that $d$ takes all values in $\mathbb{F}_q^*$ such that $\Tr_m(d^{-1})=0$. Thus $f_1$ is  affine equivalent to $G_2$
when $S\neq 0$.

Considering $f_2$,
we note that $\theta_1 = c_2^2$ must be nonzero. Since  $T = c_2$ we must also have $S = c_2^2 + c_2 = 0$. Therefore $c_2 = 1$. From the trace condition, we have $$\Tr_m\left(\frac{\theta_2^{q+1}}{\theta_1^2}\right) = \Tr_m\left( \frac{\theta_2}{\theta_1}\right) = \Tr_m(1) = 0.$$ Therefore $m$ must be even, so $\mathbb{F}_q$ contains a primitive cube root of unity $\omega$. We have 
 $$L(G_1(L(x))) = (1+\omega)f_2(x)$$ where $L(x) = x^q + \omega x$ which is a bijection since $\omega \notin U$ when $m$ is even.
 This completes the proof of Theorem~\ref{thm-small-field}.




\section{Main result}

\begin{theorem}
	\label{thm-main}
	Suppose that $m\ge 4$ is an integer and let $q=2^m$. Let $f:\F_{q^2}\ra\F_{q^2}$ be 
	given by $f(x)=x^{3q}+a_1x^{2q+1}+a_2x^{q+2}+a_3x^3$
	where $a_1,a_2,a_3\in\F_{q^2}$. If $f$ is APN, then
	$f$ is affine equivalent to $G_1(x)=x^3$ or
	$f$ is affine equivalent to $G_2(x)=x^{2^{m-1}+1}$.
\end{theorem}
\begin{proof}
	The proof of the theorem is obtained by combining 
	the results of 
	Lemma~\ref{lem-a1-small-field},
	Theorem~\ref{thm-LLHQ}, and
	Propositions \ref{prop-make-a2-zero}
	through \ref{prop-a1-a2-eq-a3q}.
\end{proof}

\begin{proposition}
	\label{prop-make-a2-zero}
	Let $f:\F_{q^2}\ra\F_{q^2}$ be 
	given by $f(x)=x^{3q}+a_1x^{2q+1}+a_2x^{q+2}+a_3x^3$
	where $a_1\in\F_{q}$, $a_2\in\F_{q^2}^*$ and  $a_3\in\F_{q^2}$.
	If $a_1/a_2\not\in U$ and $a_1/a_2\neq a_3^q$,
	then $f$ is affine equivalent to
	$h(x)= x^{3q}+a_1'x^{2q+1}+a_3'x^3$ where $a_1'\in\F_{q}$
	and $a_3'\in\F_{q^2}$.
\end{proposition}
\begin{proof}
Let $t=a_1/a_2$ and let $h_0(x)=L(f(x))$ where $L(z)=z^q+tz$.
Then
\[
h_0(x)=\left(\frac{a_1}{a_2}+a_3^q\right)x^{3q}
+\left(\frac{a_1^2}{a_2}+a_2^q\right)x^{2q+1}
+\left(\frac{a_1a_3}{a_2}+1\right)x^{3}.
\]
Since $a_1/a_2\not\in U$, functions $f$ and $h_0$ are affine equivalent.
Now let
\[
h_1(x)=\left(\frac{a_1}{a_2}+a_3^q\right)^{-1} \cdot h_0(x)
\]
and apply Lemma~\ref{lem-a1-small-field} to function $h_1$.
\end{proof}

We first deal with the special cases when two coefficients
of the function vanish.

\begin{proposition}
	\label{prop-a1-or-a3-is-zero}
	Let $f:\F_{q^2}\ra\F_{q^2}$ be 
	given by $f(x)= x^{3q}+a_1x^{2q+1}+a_3x^3$ where $a_1,a_3\in\F_{q^2}$
	and $a_1=0$ or $a_3=0$. 
	If $f$ is APN, then $f$ is affine equivalent to $G_1(x)=x^3$.
\end{proposition}
\begin{proof}
	Assume $a_1=a_2=0$. Then $\t_1=a_3^{q+1}+1$. If $a_3\in U$, then
	$\t_1=0$ and $f$ is not APN. If $a_3\not\in U$, let $L(z)=z^q+tz$
	where $t=a_3^q$. Then $L$ is a permutation of $\F_{q^2}$
	and $L(f(x))=(a_3^{q+1}+1)x^3$.
	
	Next assume $a_2=a_3=0$. Then $\t_1=(a_1+1)^2$.
	After plugging into (\ref{Eq-4}) we get $a_1^4(a_1+1)^2=0$, and
	after plugging into (\ref{Eq-5}) we get $a_1^2(a_1+1)^2=0$.
	If $f$ is APN, then $\t_1\neq 0$, hence $a_1=0$ and $f(x)=x^{3q}$
	which is affine equivalent to $G_1(x)=x^3$.
\end{proof}

\begin{proposition}
	\label{prop-Gamma1}
	Let $f:\F_{q^2}\ra\F_{q^2}$ be 
	given by $f(x)= x^{3q}+a_1x^{2q+1}+a_3x^3$ where $a_1\in\F_{q}$
	and $a_3\in\F_{q^2}$. If $(a_1,0,a_3)\in\Gamma_1$ as defined
	in equation (\ref{Eq-4}) then $f$ is affine equivalent to $G_1(x)=x^3$ or
	$f$ is affine equivalent to $G_2(x)=x^{2^{m-1}+1}$.
\end{proposition}
\begin{proof}
	We will show that $a_3\in\Fq$ or $a_1=0$
	and then the conclusion will follow
	by applying Theorem~\ref{thm-small-field}
	or Proposition~\ref{prop-a1-or-a3-is-zero} accordingly.
	
    If $a_1=0$ or $a_3=0$, then the conclusion follows immediately.
    Hence assume that $a_1\neq 0$ and $a_3\neq 0$.
	Let $a_3=yz$ where $y\in\F_{q}^*$ 
	and $z\in U$. After plugging this along with $a_2=0$ into equation (\ref{Eq-4}), dividing by $a_1^2$ and simplifying we get
	\begin{equation}
	\label{eq4-new}
		a_1yz^2+(a_1^4+a_1^2+y^4+y^2)z+a_1y=0.
	\end{equation}
	By the definition of set $\Gamma_1$ we have 
	\[
	\Tr_m \left( \frac{\t_2^{q+1}}{\t_1^2}\right)
	=
	\Tr_m \left( \frac{a_1^2}{(a_1^2+y^2+1)^2}\right)
	=
	\Tr_m \left( \frac{a_1}{a_1^2+y^2+1}\right)
	=0.
	\]
	Hence there exist $t\in\F_q$ such that
	\begin{equation}
	\label{eq-G1-t}
	(t^2+t)(a_1^2+y^2+1)+a_1=0.
	\end{equation}
	The resultant of the left-hand sides
	of (\ref{eq4-new}) and (\ref{eq-G1-t}) with respect to $y$
	is $a_1^2(Az^2+B)(Bz^2+A)$ where $A=(a_1t+t+1)t^3$
	and $B=(a_1t+a_1+t)(t+1)^3$. 
	
	If $z=1$, then $a_3\in\Fq$. Otherwise,
	we have $A,B\in\F_q$ but $z\not\in\F_q$. Then the resultant
	can vanish only if $A=B=0$. This occurs only if $(a_1,t)=(0,0)$
	or $(a_1,t)=(0,1)$. This completes the proof.
\end{proof}

\begin{proposition}
	\label{prop-Gamma2}
	Let $f:\F_{q^2}\ra\F_{q^2}$ be 
	given by $f(x)= x^{3q}+a_1x^{2q+1}+a_3x^3$ where $a_1\in\F_{q}$
	and $a_3\in\F_{q^2}$. If $(a_1,0,a_3)\in\Gamma_2$ as defined
	in equation (\ref{Eq-5}) then $f$ is affine equivalent to $G_1(x)=x^3$ or
	$f$ is affine equivalent to $G_2(x)=x^{2^{m-1}+1}$.
\end{proposition}
\begin{proof}
 We will show that $a_3\in\F_q$ and then the conclusion will follow
 by applying Theorem~\ref{thm-small-field}.
 
 Again let $a_3=yz$ where $y\in\F_{q}$ 
 and $z\in U$. After plugging this along with $a_2=0$ into equation (\ref{Eq-5}) and simplifying we get
 \begin{equation*}
 (a_1^2y)z^2+(a_1^3+a_1y^2+a_1)z+a_1^4y+a_1^2y+y^5+y=0.
 \end{equation*}
 If $z=1$ then $a_3\in\F_q$ and we are done. Otherwise, by Lemma~\ref{lem-quad-U} we must have
 \[
 (a_1^2y)+(a_1^4y+a_1^2y+y^5+y)=y(a_1+y+1)^4=0.
 \]
 If $y=0$, then $a_3\in\F_q$. If $(a_1+y+1)^4=0$, then
 $\t_1=(a_1+y+1)^2=0$ and $(a_1,0,a_3)\not\in\Gamma_2$
 by the first condition of equation (\ref{Eq-5}).
\end{proof}

\begin{proposition}
	\label{prop-a1-a2-in-U}
	Let $f:\F_{q^2}\ra\F_{q^2}$ be 
	given by $f(x)= x^{3q}+a_1x^{2q+1}+a_2x^{q+2}+a_3x^3$ where $a_1\in\F_{q}$, $a_2\in\F_{q^2}^*$ and  $a_3\in\F_{q^2}$, and assume that $a_1/a_2\in U$.
	If $f$ is APN, 
	then $m$ is even, and 
	furthermore
	$a_1,a_2,a_3\in\F_q$ or
	there exist $u,z\in U$ such that
	\begin{eqnarray*}
	a_1 & = & \frac{(u^3+z)^2}{u(u^2+z)^2} \\
	a_2 & = & \frac{(u^3+z)^2}{(u^2+z)^2}  \\
	a_3 & = & \frac{uz^2(u+1)^2}{(u^2+z)^2}.
	\end{eqnarray*}
\end{proposition}
\begin{proof}
	Under the assumptions listed in the proposition, we first prove that
	$(a_1,a_2,a_3)\not\in\Gamma_1$, and then we prove that
	$(a_1,a_2,a_3)\in\Gamma_2$ implies the conclusions given in the statement.
	
    From $a_1/a_2\in U$ it follows that $a_1\neq 0$.
	
	Let $a_2=a_1u$ where $u\in U$, and let $a_3=yz$ where $y\in\F_{q}$ 
	and $z\in U$. After substituting this into (\ref{Eq-4}),
	simplifying and clearing denominators we get
	\[
	a_1^2(a_1u^4y+a_1u^3z+u^2y^2z+a_1yz^2+a_1uz+u^2z)(yz+u)(uy+z)=0.
	\]
	Since $u/z\in U$, any of the last two factors can vanish only
	if $y=1$, but then $\t_1=y^2+1=0$ and $(a_1,a_2,a_3)\not\in\Gamma_1$.
	Since $a_1\neq 0$, the second factor must vanish, hence
	\[
	a_1(u^4y+u^3z+yz^2+uz)  =  u^2z(y^2+1).
	\]
	The right-hand side can not vanish, hence we have
	\begin{equation}
	\label{eq-a1sol-P27}
	a_1=\frac{u^2z(y^2+1)}{u^4y+u^3z+yz^2+uz}.
	\end{equation}
	
	We will obtain the desired conclusion by invoking the condition\break
	$\tr_{m}\left(\frac{\theta_2^{q+1}}{\theta_1^2}\right) =0$
	from the definition of set $\Gamma_1$ in equation (\ref{Eq-4}).
	After substituting the formula (\ref{eq-a1sol-P27}) for $a_1$
	into $\t_2$ we get
	\[
	\frac{\theta_2^{q+1}}{\theta_1^2}=
	\frac{u^3z(yz+u)(uy+z)}{(u^4y+u^3z+yz^2+uz)^2}.
	\]
	For the trace value condition to hold there must exist $w\in\F_q$
	such that
	\[
	(w^2+w)(u^4y+u^3z+yz^2+uz)^2+u^3z(yz+u)(uy+z)=0.
     \] 
     This equation factors as
     \[
     (Dw+z(yz+u))(Dw+u^3(uy+z))=0
     \]
     where $D=u^4y+u^3z+yz^2+uz$ is the denominator in (\ref{eq-a1sol-P27}).
     We already know that $D\neq 0$. Hence the values of $w$ are
     \[
     w_1=\frac{z(yz+u)}{D}, \ \ \  w_2=\frac{u^3(uy+z)}{D}.
     \] 
     Let
     \[
     r=\frac{a_1}{w_1}=\frac{u^2(y^2+1)}{yz+u}.
     \]
     Since $a_1,w_1\in \F_q$, we must have $r\in\F_q$. However
     \[
     r+r^q=\frac{(y+1)^2D}{u(yz+u)(uy+z)}\neq 0
     \]
     which finishes the proof that $(a_1,a_2,a_3)\not\in \Gamma_1$.
     
     Thus, if $f$ is APN, then
     we must have
     $(a_1,a_2,a_3)\in \Gamma_2$ and $m$ must be even.
     We will analyze this case now.
     
     Again assume that $a_2=a_1u$ and $a_3=yz$ where $y\in\F_q$
     and $u,z\in U$.
     By Proposition~\ref{prop-a1-or-a3-is-zero} 
     we can assume $a_3\neq 0$, hence $y\neq 0$.
     After plugging the expressions for $a_2$ and $a_3$ into (\ref{Eq-5})
     and simplifying we get
     \[
     a_1(uy+z)(Aa_1^2+Ba_1+C)=0
     \] 
     where $A=D(yz+u)$, $D=(u^2+z)^2y+uz(u^2+1)$, 
     $B=(y+1)^2u^3(uy+z)$ and
     $C=(y+1)^4u^2z$.
     Since $\t_1=1+y^2$, we have $y\neq 1$, hence
     $uy+z\neq 0$ and $yz+u\neq 0$.
     
     We will distinguish two cases: $A=0$ and $A\neq 0$,
     equivalently, $D=0$ and $D\neq 0$.
     
     First let us assume that $D=0$.     
     If $u^2=z$, then $u=z=1$ and $a_1,a_2,a_3\in\F_q$.
     Otherwise
      \[
     y=\frac{uz(u^2+1)}{(u^2+z)^2}
     \]
     and
     \[
     a_1=\frac{C}{B}=\frac{(y+1)^2z}{u(uy+z)}
     =\frac{(u^3+z)^2}{u(u^2+z)^2}.
     \]
     Then 
     \[
     a_2=a_1u=\frac{(u^3+z)^2}{(u^2+z)^2}
     \]
     and
     \[
     a_3=yz=\frac{uz^2(u+1)^2}{(u^2+z)^2}.
     \]
     With these values of $a_1,a_2,a_3$ we get after simplifications
     $\t_2^{q+1}/\t_1^2=1$, hence $\Tr_m(\t_2^{q+1}/\t_1^2)=0$ since $m$ is even,
     and $(a_1,a_2,a_3)\in\Gamma_2$.
     
     We complete the proof by analyzing the case $A\neq 0$.
     We note that $a_1$ is also a root of 
     \[
     Aa_1^2+Ba_1+C+\frac{A}{A^q}(A^qa_1^2+B^qa_1+C^q)=0,
     \]
     which after simplifications becomes
     $
     u(Ra_1+S)=0
     $
     where\\
     $
     R=(u+z)^4(y+1)^2y^2
     $ and $S=(y+1)^4(uyz^2+u^2z+uy+z)uz$. If $u=z$, then it must be
     $S=z^3(y+1)^5(z+1)^2=0$. Since $\t_1=(y+1)^2$, we get $u=z=1$ and $a_1,a_2,a_3\in\F_q$. If $y=0$, then we must have 
     $S=z^2u(u+1)^2=0$, hence $u=1$ and $a_1,a_2,a_3\in\F_q$.
     Otherwise we get
     \begin{equation}
     \label{eq-a1-dead-end}
     a_1=\frac{S}{R}
     =\frac{(y+1)^2(uyz^2+u^2z+uy+z)uz}{(u+z)^4y^2}.
     \end{equation}
    After plugging this into (\ref{Eq-5}) we get
    \begin{equation}
    \label{eq-long-eq}
    \frac{z^2(y+1)^4u^2(u+1)^2(u+yz)((u+z^2)y + z(u+1) )^2D}
    {(y(u+z)^2 )^4}=0.
    \end{equation}
    
    If $u=1$ then
    \[
    a_1=\frac{z(y+1)^2}{y(z+1)^2}
    \]
    and
    \[
    \frac{\t_2^{q+1}}{\t_1^2}=
    \frac{(yz+1)(y+z)z}{y^2(z+1)^4}.
    \]
    For the trace condition to be satisfied there must exist
    $w\in\F_q$ such that
    \begin{equation}
    \label{eq-if-u-1}
    (w^2+w)(y^2(z+1)^4)+(yz+1)(y+z)z=0
    \end{equation}
    which factors as
    \[
    (wyz^2 + wy + y + z)  (wyz^2 + wy + yz^2 + z)=0
    \]
    and the root of the first factor is
    \[
    w_1=\frac{y+z}{y(z^2+1)}
    \]
    for which we get
    \[
    w_1+w_1^q = 1.
    \]
    Hence $w_1\not\in\F_q$. Of course for the other root
    of (\ref{eq-if-u-1}) we have $w_2=w_1+1\not\in\F_q$.
    Therefore the case $u=1$ can not occur.
    
    Otherwise, $u\neq 1$ and since we are also assuming $D\neq 0$,
    we get from (\ref{eq-long-eq})
    \[
    y=\frac{z(u+1)}{u+z^2}
    \]
    and after plugging this into (\ref{eq-a1-dead-end}) we get
    \[
    a_1=\frac{u(z+1)^2}{(u+1)(u+z^2)}
    \]
    and after plugging the last two formulas into $\t_2^{q+1}/\t_1^2$
    and simplifying we get
    \[
    \frac{\t_2^{q+1}}{\t_1^2}=\frac{u}{(u+1)^2}.
    \]
     Since we require $\Tr_m(\t_2^{q+1}/\t_1^2)=0$,
     there must exist $w\in\F_q$ such that
    $(w^2+w)(u+1)^2+u=0$. This factors as
    \[
    (uw+w+1)(uw+u+w)=0
    \]
    and the roots are $w_1=1/(u+1)$ and $w_2=u/(u+1)$.
    However $w_1+w_1^q=w_2+w_2^q=1$ which means
    that $w_1,w_2\not\in\F_q$ and $\Tr_m(\t_2^{q+1}/\t_1^2)\neq 0$.
    Hence the case $A\neq 0$ does not produce APN functions,
    and the entire proof of the proposition is now completed.
\end{proof}

\begin{proposition}
	\label{prop-G2-equiv}
	Let $f:\F_{q^2}\ra\F_{q^2}$ be 
	given by $f(x)= x^{3q}+a_1x^{2q+1}+a_2x^{q+2}+a_3x^3$ where
	\begin{eqnarray*}
		a_1 & = & \frac{(u^3+z)^2}{u(u^2+z)^2} \\
		a_2 & = & \frac{(u^3+z)^2}{(u^2+z)^2}  \\
		a_3 & = & \frac{uz^2(u+1)^2}{(u^2+z)^2}
	\end{eqnarray*}
    for some $u,z\in U$ such that $u^2\neq z$, and $m$ is even.
    If $f$ is APN, then $f$ is affine equivalent to $G_1(x)=x^3$
    or 
    $f$ is affine equivalent to $G_2(x)=x^{2^{m-1}+1}$.
\end{proposition}
\begin{proof}
	Since $m$ is even, $\F_q$ contains a primitive cube root of unity,
	which we will denote $\w$, that is, $\w^2+\w+1=0$.

Let $D=(\w^2u^3+(u+\w) z)^2$.
Let $L_1(x)=x^q+tx$ where
\begin{equation}
\label{eq-def-t}
t=\frac{(\w u^3 + (u+\w^2)z)^2}{uD}
\end{equation}
and let $L_2(x)=x^{2q}+sx^2$ where $s=\w u^2$.

Let us consider the cases when the denominator of (\ref{eq-def-t}) may
vanish. This happens when $z_1\in U$ where
\[
z_1=\frac{w^2 u^3}{u+\w }.
\]
We have
\[
\frac{1}{z_1}+z_1^q=\frac{\w(u+1)^2}{u^3(\w+u+1)}
\]
hence $z_1\in U$ only if $u=1$, but then $z_1=1$ and $u^2=z_1$, hence
this case can not occur.

Let $G_2(x)=x^{2^{m-1}+1}$ be a Gold APN function.
Denote
\begin{equation}
\label{eq-def-c0123}
F(x)=L_1(G_2(L_2(x))) = c_0x^{3q}+c_1x^{2q+1}+c_2x^{q+2}+c_3x^3.
\end{equation}
A calculation shows that
\begin{eqnarray*}
c_0 & = & \frac{(u^2+z)^2}{D}
\\
c_1 & = & \frac{(u^3+z)^2}{uD}
\\
c_2 & = & \frac{(u^3+z)^2}{D}
\\
c_3 & = & \frac{u z^2 (u+1)^2}{D}.
\end{eqnarray*}
Since $u^2\neq z$, we have $c_0\neq 0$.
Since $c_i/c_0=a_i$ for all $i=1,2,3$, we get that $F(x)=c_0f(x)$.

We need to show that $L_1$ and $L_2$ are permutations of $\F_{q^2}$.
For $L_1$ this is the case if and only if $t\not \in U$.
We can assume without loss of generality that $t\neq 0$ and compute
\[
t^q+\frac{1}{t}=\frac{u^3(u^3+z)^2(u+z)^2}
            {\w(u^3+\w^2uz+\w z)^2  (u^3+\w^2 u^2+\w z)^2 }.
\]
Hence $t^{q+1}\neq 1$ unless $u^3=z$ or $u=z$.
If $u^3=z$ then $a_1=a_2=0$ and $f$ is affine equivalent
to $G_1(x)=x^3$ by Proposition~\ref{prop-a1-or-a3-is-zero}.
If $u=z$ then $\t_1=0$ and $f$ is not APN.

To show that $L_2$ is a permutation of $\F_{q^2}$ we need to
show that the only root of $L_2(x)=x^2(x^{2q-2}+s)$ 
in $\F_{q^2}$
is $0$, equivalently,
that there does not exist $x\in\F_{q^2}$ such that $x^{2q-2}=s$.
After raising this equation to power $q+1$ the left-hand side
becomes $(x^{2q-2})^{q+1}=(x^{q^2-1})^2=1$ but the right-hand side
becomes $s^{q+1}=(\w u^2)^{q+1}=\w^{q+1}=\w^2$ since $m$ is even.
Hence the equation can not hold, and $L_2$ is a permutation of $\F_{q^2}$. 
\end{proof}

\begin{proposition}
	\label{prop-a1-a2-eq-a3q}
	Let $f:\F_{q^2}\ra\F_{q^2}$ be 
	given by $f(x)= x^{3q}+a_1x^{2q+1}+a_2x^{q+2}+a_3x^3$ where $a_1\in\F_{q}$, $a_2\in\F_{q^2}^*$ and  $a_3\in\F_{q^2}$, and assume that $a_1/a_2=a_3^q$.
	If $f$ is APN, then $f$ is affine equivalent to $G_1(x)=x^3$.
\end{proposition}
\begin{proof}
	Assume $a_1=a_2a_3^q$. There exist $v,y\in\F_q$ and $z\in U$
	such that $a_1=vy$, $a_2=vz$ and $a_3=yz$. 
	After simplifications we get $\t_1=(v+1)^2(y+1)^2$.
	After plugging
	into (\ref{Eq-4}) and simplifications we get
	\[
	v^2(v+1)^4(y+1)^6=0.
	\]
	Since $\t_1\neq 0$, we must have $v=0$. Then $a_1=a_2=0$ and
	the conclusion follows from Proposition~\ref{prop-a1-or-a3-is-zero}.
	Finally after plugging  into (\ref{Eq-5}) and simplifications we get
	\[
	\frac{v(v+1)^4(y+1)^6}{z}=0
	\]
	and the conclusion is the same as in the previous case.
\end{proof}

\end{document}